\documentclass[12pt]{article}

\usepackage{amsmath,amssymb,amsthm}
\usepackage{geometry}
\usepackage{tikz,tkz-graph,tikz-cd,tikz-qtree,youngtab}
\usepackage{array,booktabs}

\newtheorem{theorem}{Theorem}[section]
\newtheorem{lemma}[theorem]{Lemma}
\newtheorem{corollary}[theorem]{Corollary}

\newtheorem{proposition}[theorem]{Proposition}

\theoremstyle{remark}
\newtheorem{remark}[theorem]{Remark}
\theoremstyle{definition}
\newtheorem{example}[theorem]{Example}
\theoremstyle{definition}
\newtheorem{definition}[theorem]{Definition}

\begin{document}
	
\bibliographystyle{abbrv}
	
\title{The Arithmetic Singleton Bound on the Hamming Distances of Simple-root Constacyclic Codes over Finite Fields}
\author{Li Zhu$^{1}$ and Hongfeng Wu$^{2}$\footnote{Corresponding author.}
\setcounter{footnote}{-1}
\footnote{E-Mail addresses:
lizhumath@pku.edu.cn (L. Zhu), whfmath@gmail.com (H. Wu)}
\\
{1.~School of Mathematical Sciences, Guizhou Normal University, Guiyang, China}
\\
{2.~College of Science, North China University of technology, Beijing, China}}
	
\date{}
\maketitle
	
\thispagestyle{plain}
\setcounter{page}{1}	
	
\begin{abstract}
In this work, We introduce a new upper bound on the Hamming distance of simple-root constacyclic codes over finite fields, which we call the arithmetic Singleton bound. The main technical tool is the notion of a multiple equal-difference (MED) representation. Via the MED representations of the defining set of the generator polynomial of a simple-root constacyclic code, we obtain a family of upper bounds on its Hamming distance, among which the weakest one coincides with the Singleton bound, while the strongest one is defined to be the arithmetic Singleton bound for this code. Consequently, the arithmetic Singleton bound is always at least as strong as the classical Singleton bound, and is in fact strictly stronger in numerous nontrivial cases. The arithmetic Singleton bound partially measures the restriction on the Hamming distance of a simple-root constacyclic code imposed by its arithmetic structure. In particular, for an irreducible constacyclic code, the MED representations of the defining set of its generator polynomial are completely determined, via which the arithmetic Singleton bound is computed concretely. Finally for any simple-root cyclic code the arithmetic Singleton bound and the BCH bound are compared.
\end{abstract}

\section{Introduction}\label{sec 1}
Let $\mathbb{F}_{q}$ be a finite field with $q$ elements. A linear $[m,k,d]$-code over $\mathbb{F}_{q}$ is a $k$-dimensional 
subspace of $\mathbb{F}_{q}^{m}$ with Hamming distance $d$. The Hamming distance is of great significance in coding theory as 
it controls the ability of the code to correct errors. The Singleton bound is one of the most well-known upper bound on 
the Hamming distance of a linear code, which asserts that the parameters $m$, $k$ and $d$ should satisfy the inequality
$$d \leq m-k+1.$$

Constacyclic codes are a subclass of linear codes which are equipped with additional structure of an ideal in a quotient ring of polynomials. Concretely, for a nonzero element $\lambda$ in $\mathbb{F}_{q}$ and a positive integer $m$, a $\lambda$-constacyclic code of length $m$ over $\mathbb{F}_{q}$ can be naturally identified with an ideal in the ring $\mathbb{F}_{q}[X]/(X^{m}-\lambda)$. Compared with the linear structure alone, we will refer to the structure of an ideal in a polynomial quotient ring of constacyclic codes as their arithmetic structure. The arithmetic structure, on one hand, endows constacyclic codes with advantages in both the theoretical and the applied aspects; for instance constacyclic codes have efficient encoding and decoding algorithms and they can be applied to construct various other kinds of codes, while on the other hand, there is common belief that the arithmetic structure imposes stronger restriction on the Hamming distances of constacyclic codes than the linear structure purely does. Much empirical evidence supports this belief(See for instance \cite{Alderson}, \cite{Huffman}, \cite{MacWilliams} et al..). However, there is currently no rigorous and systematic theory which measures such restriction explicitly and also demonstrate the mechanism through which the restriction happens.

This paper aims to partially establish this theoretical framework. The central tool which to some extent captures the constraint on the Hamming distance of a constacyclic code imposed by its arithmetic structure is the notion of a multiple equal-difference representation (a MED representation) of the defining set of a squarefree polynomial. Let $f(X)$ be a squarefree polynomial over $\mathbb{F}_{q}$ with order $n$ and defining set $T_{f}$. A MED representation of $T_{f}$ is a partition
$$T_{f} = \bigsqcup_{i=1}^{s}E_{i}$$
of $T_{f}$ into some cosets $E_{i}$ with respect to an additive subgroup $d\mathbb{Z}/n\mathbb{Z} \subseteq \mathbb{Z}/n\mathbb{Z}$. The MED representations of $T_{f}$ describe the symmetry of $T_{f}$, or equivalently, the symmetry of the set of roots of $f(X)$. By Proposition \ref{prop 1}, which gives a one-to-one correspondence between the equal-difference subsets of $\mathbb{Z}/n\mathbb{Z}$ with the binomial factors of $f(X)$, the symmetry of $T_{f}$ can be translated into results on the coefficients of $f(X)$, in particular, on the number of nonzero coefficients of $f(X)$. Consider the simple-root $\lambda$-constacyclic code of length $m$
$$\mathcal{C} = (f(X)) \subseteq \mathbb{F}_{q}[X]/(X^{m}-\lambda),$$
where $X^{m}-\lambda$ is a multiple of $f(X)$. The MED representations of $T_{f}$ give rise to a family of upper bounds on the Hamming distance of $\mathcal{C}$, among which the weakest one coincides with the Singleton bound. As the Singleton bound can be obtained from the linear structure of $\mathcal{C}$, the strongest bound lying in this family can be viewed to partially measure the constraint on the Hamming distance of $\mathcal{C}$ imposed by its arithmetic structure. For this reason we refer to it as the arithmetic Singleton bound.

By its definition the arithmetic Singleton bound is always stronger than the Singleton bound. Moreover, the criterion for these two bounds coinciding, and hence for the arithmetic Singleton bound being strictly stronger than the Singleton bound, is obtained, from which one sees that there are infinitely many nontrivial families of constacyclic codes for which the arithmetic Singleton bound is strictly stronger than the Singleton bound.

In particular, for any irreducible constacyclic code the arithmetic Singleton bound, along with its difference from the Singleton bound, can be obtained concretely. Let $\mathcal{C}$ be a constacyclic code generated by an irreducible polynomial $f(X)$ over $\mathbb{F}_{q}$, and let $n$ be the order of $f(X)$. We prove that the arithmetic Singleton bound is determined by $\mathrm{ord}_{\mathrm{rad}(n)}(q)$, up to a constant factor which is either $1$ or $2$. Therefore if the prime divisors of $n$ are fixed, the arithmetic Singleton bound is bounded, and consequently the difference of the arithmetic Singleton bound and the Singleton bound grows dramatically when $n$ grows.

The paper is organized as follows. Section \ref{sec 2} is devoted to recalling some basic definitions and results, and to 
fixing notations. In Section \ref{sec 3} we generalize the definition of a MED representation of a cyclotomic coset, which is 
introduced in \cite{Zhu3}, to the general case of the defining set $T_{f}$ of a squarefree polynomial $f(X)$ over 
$\mathbb{F}_{q}$ with a nonzero constant term. The structure of the space $\mathcal{MER}(T_{f})$ of the MED representations of 
$T_{f}$ is characterized. In Section \ref{sec 4} we establish the arithmetic Singleton bound. We prove that each MED 
representation of the defining set of the generator polynomial $f(X)$ of a simple-root constacyclic code $\mathcal{C}$ induces 
an upper bound on the Hamming distance of $\mathcal{C}$, among which the weakest one coincides with the Singleton bound and the 
strongest one is defined to be the arithmetic Singleton bound.
The criterion for the arithmetic Singleton bound being equal to the classical Singleton bound is obtained. In the case where 
$\mathcal{C}$ is irreducible, the arithmetic Singleton bound for $\mathcal{C}$ is computed explicitly. In particular, we show that 
the arithmetic Singleton bound for $\mathcal{C}$ is completely determined by $q$ and the order of $f(X)$. Finally in Section \ref{sec 5} the arithmetic Singleton bound and the BCH bound for simple-root cyclic codes are compared. An equivalent condition on that the arithmetic Singleton bound coincides with the BCH bound, in which case the Hamming distance can be obtained directly, is presented. As corollaries, an equivalent condition for an irreducible cyclic code to be of Hamming distance $2$, and a sufficient condition for it to be of Hamming distance $3$ are given.
	
\section{Preliminaries}\label{sec 2}
\subsection{Basic number theory}
Throughout this paper, it is assumed that $p$ is a prime number, $q=p^{e}$ is a power of $p$, and $n$ is a positive integer not divided by $p$. Denote by $\mathbb{Z}/n\mathbb{Z}$ the ring of residue classes modulo $n$. As there is a canonical surjection
$$\mathbb{Z} \rightarrow \mathbb{Z}/n\mathbb{Z}: \quad \gamma \mapsto \overline{\gamma} = \gamma+n\mathbb{Z},$$
if making no confusion, we identify $\gamma$ with its image $\overline{\gamma}$. Therefore the notation $\gamma \in \mathbb{Z}/n\mathbb{Z}$ actually stands for the image of $\gamma$ in $\mathbb{Z}/n\mathbb{Z}$.

Assume that $n$ is factorized into a product of prime powers as $n = p_{1}^{e_{1}}\cdots p_{s}^{e_{s}}$, where $p_{1},\cdots,p_{s}$ are distinct primes, and $e_{1},\cdots,e_{s}$ are positive integers. The radical of $n$ is defined by $\mathrm{rad}(n) = p_{1}\cdots p_{s}$.

If $m$ and $n$ are coprime integers, we denote by $\mathrm{ord}_{n}(m)$ the order of $m$ in the multiplicative group $(\mathbb{Z}/n\mathbb{Z})^{\ast}$, i.e., the smallest positive integer such that 
$$m^{\mathrm{ord}_{n}(m)} \equiv 1 \pmod{n}.$$
It is clear that $\mathrm{ord}_{n}(m)$ divides the order $\phi(n)$ of $(\mathbb{Z}/n\mathbb{Z})^{\ast}$.

Let $\ell$ be a prime. Denote by $v_{\ell}(n)$ the $\ell$-adic valuation of $n$, i.e., the maximal integer such that $\ell^{v_{\ell}(n)} \mid n$. The following lift-the-exponent lemmas are well-known.

\begin{lemma}{\cite{Nezami}}\label{lem 3}
	Let $\ell$ be an odd prime number, and $m$ be an integer such that $\ell \mid m-1$. Then $v_{\ell}(m^{d}-1) = v_{\ell}(m-1) + v_{\ell}(d)$ for any positive integer $d$.
\end{lemma}

\begin{lemma}{\cite{Nezami}}\label{lem 2}
	Let $m$ be an odd integer, and $d$ be a positive integer.
	\begin{description}
		\item[(1)] If $m \equiv 1 \pmod{4}$, then
		$$v_{2}(m^{d}-1) = v_{2}(m-1) + v_{2}(d), \  v_{2}(m^{d}+1) = 1.$$
		\item[(2)] If $m \equiv 3 \pmod{4}$ and $d$ is odd, then
		$$v_{2}(m^{d}-1) = 1, \  v_{2}(m^{d}+1) = v_{2}(m+1).$$
		\item[(3)] If $m \equiv 3 \pmod{4}$ and $d$ is even, then
		$$v_{2}(m^{d}-1) = v_{2}(m+1) + v_{2}(d), \  v_{2}(m^{d}+1) = 1.$$
	\end{description}
\end{lemma}

\subsection{Finite fields}
Let $\mathbb{F}_{q}$ be a finite field containing $q$ elements, and $\mathbb{F}_{q}^{\ast}$ be the multiplicative group  of nonzero elements in $\mathbb{F}_{q}$. For any $\lambda \in \mathbb{F}_{q}^{\ast}$, the order $\mathrm{ord}(\lambda)$ of $\lambda$ is defined to be the smallest positive integer $n$ such that $\lambda^{n}=1$. It is clear that $\mathrm{ord}(\lambda)$ divides $q-1$.

Given a positive integer $n$ not divisible by $p$, there are $n$ roots of $X^{n}-1$ lying in some finite extension of $\mathbb{F}_{q}$. A root $\zeta_{n}$ of $X^{n}-1$ that is not a root of $X^{m}-1$ for any $m < n$ is called a primitive $n$-th root of unity. Throughout this paper we fix a family 
$$\{\zeta_{n} \ | \ \mathrm{gcd}(n,q)=1\}$$
of primitive roots of unity, which is compatible in the sense that for any $m,n \in \mathbb{N}^{+}$ such that $m \mid n$ and $\mathrm{gcd}(n,q)=1$ it holds that 
$$\zeta_{n}^{\frac{n}{m}} = \zeta_{m}.$$

In this paper, unless stating otherwise, $f(X)$ always denotes a nonconstant squarefree monic polynomial over $\mathbb{F}_{q}$, with a nonzero constant term. There is a smallest positive integer $n$ such that $f(X) \mid X^{n}-1$. The integer $n$ is called the order of $f(X)$ and is denoted by $n=\mathrm{ord}(f)$. If $f(X)$ is irreducible over $\mathbb{F}_{q}$, then $\mathrm{ord}(f) = \mathrm{ord}(\delta)$ for any root $\delta$ of $f(X)$ lying in the finite extension field $\mathbb{F}_{q}[X]/(f(X))$ of $\mathbb{F}_{q}$.

Assume that the order of $f(X)$ is $n$. Then $f(X)$ can be written as 
$$f(X) = \prod_{\gamma \in T_{f}}(X - \zeta_{n}^{\gamma})$$
for some $T_{f} \subseteq \mathbb{Z}/n\mathbb{Z}$. This set $T_{f}$ is referred to as the defining set of $f(X)$.

Let $n$ be a positive integer coprime to the prime power $q$. For any $\gamma \in \mathbb{Z}/n\mathbb{Z}$, the $q$-cyclotomic coset modulo $n$ containing $\gamma$ is defined to be 
$$c_{n/q}(\gamma) = \{\gamma,\gamma q,\cdots,\gamma q^{\tau-1}\} \subseteq \mathbb{Z}/n\mathbb{Z},$$
where $\tau$ is the smallest positive integer such that $\gamma q^{\tau}\equiv\gamma\pmod n$. 

It is well known that there is a one-to-one correspondence between the $q$-cyclotomic cosets modulo $n$ and the irreducible factors of $X^{n}-1$ over $\mathbb{F}_{q}$. Fixing a primitive $n$-th root of unity $\zeta_{n}$, we assigns to $c_{n/q}(\gamma) \in \mathcal{C}_{n/q}$ the irreducible polynomial
$$M_{\gamma}(X) = (X-\zeta_{n}^{\gamma})(X-\zeta_{n}^{\gamma q})\cdots(X-\zeta_{n}^{\gamma q^{\tau-1}}).$$
It is straightforward to verify if $\gamma$ is coprime to $n$ then $c_{n/q}(\gamma)$ is the defining set of $f(X)$; otherwise, setting $n_{\gamma} = \frac{n}{\mathrm{gcd}(n,\gamma)}$ and $\widetilde{\gamma} = \frac{\gamma}{\mathrm{gcd}(n,\gamma)}$, then $c_{n_{\gamma}/q}(\widetilde{\gamma})$ is the defining set of $f(X)$. We call $c_{n_{\gamma}/q}(\widetilde{\gamma})$ the primitive form of $c_{n/q}(\gamma)$.

\subsection{Constacyclic codes over finite fields}
Let $m$ be a positive integer. A linear $[m,k,d]$-code over $\mathbb{F}_{q}$ is a $k$-dimensional subspace $\mathcal{C}$ of $\mathbb{F}_{q}^{m}$, with the Hamming distance $d$. Here for a codeword $c = (c_{0},c_{1},\cdots,c_{m-1}) \in \mathcal{C}$, the Hamming weight $\mathrm{wt}_{H}(c)$ of $c$ is the number of the nonzero terms among $c_{0},c_{1},\cdots,c_{m-1}$, and the Hamming distance $d_{H}(\mathcal{C})$ of $\mathcal{C}$ is the minimal Hamming weight of any nonzero codeword in $\mathcal{C}$. The Hamming distance of $\mathcal{C}$ plays a fundamental role as it captures the ability of $\mathcal{C}$ to correct errors. The Singleton bound is one of the most classical upper bound for the Hamming distance of a linear code, which is given below.

\begin{theorem}
	Let $\mathcal{C}$ be a linear $[m,k,d]$-code. Then
	$$d \leq m-k+1.$$
\end{theorem}

For any $\lambda \in \mathbb{F}_{q}^{\ast}$, the $\lambda$-constacyclic shift on $\mathbb{F}_{q}^{m}$ is given by
$$\tau_{\lambda}(c_{0},c_{1},\cdots,c_{m-1}) = (\lambda c_{m-1},c_{0},\cdots,c_{m-2}).$$

\begin{definition}
	A linear $[m,k,d]$-code $\mathcal{C}$ is said to be a $\lambda$-constacycic code if $\tau_{\lambda}(\mathcal{C}) = \mathcal{C}$. In particular, $\mathcal{C}$ is called cyclic if $\lambda = 1$, and is called negacyclic if $\lambda = -1$.
\end{definition}

As the $\mathbb{F}_{q}$-spaces $\mathbb{F}_{q}^{m}$ and $\mathbb{F}_{q}[X]/ (X^{m}- \lambda)$ are isomorphic, each word $c = (c_{0},c_{1},\cdots,c_{m-1})$ in a linear code $\mathcal{C}$ can be identified with a unique polynomial
$$c(X) = c_{0} + c_{1}X + \cdots + c_{m-1}X^{m-1} \in \mathbb{F}_{q}[X]/ (X^{m}- \lambda),$$
which is called the polynomial representation of $c$. Under this identification, the linear code $\mathcal{C}$ is $\lambda$-constacyclic if and only if it is an ideal of $\mathbb{F}_{q}[X]/ (X^{m}- \lambda)$. A $\lambda$-constacyclic code $\mathcal{C}$ of length $m$ is said to be simple-root if $p \nmid m$, while is said to be repeated-root if $p \mid m$.

Under the above identification, the Hamming weight of a codeword $c \in \mathcal{C}$ is exactly the number of nonzero coefficients of $c(X)$. Following this fact, we call the number of nonzero coefficients of any polynomial $f(X) \in \mathbb{F}_{q}[X]$ the Hamming weight of $f(X)$, and denote it by $\mathrm{wt}_{H}(f)$.

Let $\mathcal{C}$ be a simple-root cyclic code of length $n$ over $\mathbb{F}_{q}$, with generator polynomial $f(X)$. Assume that the order of $f(X)$ is $n$, and denote by $T_{f} \subseteq \mathbb{Z}/n\mathbb{Z}$ the defining set of $f(X)$. We say that $T_{f}$ contains $\delta$ consecutive integers for some positive integer $\delta$, if there is some integer $b$ such that $\{b,b+1,\cdots,b+\delta-1\} \subseteq T_{f}$. The BCH bound for $\mathcal{C}$ can be stated as follows.

\begin{theorem}
	Let the notations be given as above. If $T_{f}$ contains $\delta$ consecutive integers, then $d_{H}(\mathcal{C}) \geq \delta+1$.
\end{theorem}

\section{Multiple equal-difference representations of the defining set of a polynomial}\label{sec 3}
This section is devoted to defining a multiple equal-difference representation of the defining set of a squarefree polynomial with a nonzero constant term and investigating its basic properties, which plays a central role in the construction of arithmetic Singleton bound. This definition is a natural generalization of the notion of a multiple equal-difference representation of a cyclotomic coset introduced in \cite{Zhu3}. Therefore we will begin by recalling some facts in the case of a cyclotomic coset, especially the link of the multiple equal-difference representations of a cyclotomic coset $c_{n/q}(\gamma)$ with the irreducible factorizations of the induced polynomial $M_{\gamma}(X)$ into a product of binomials of the same degree over extension fields of $\mathbb{F}_{q}$, and then generalize these results to the general case. For detailed proofs of the results listed in Section \ref{sec 3.1} we refer the readers to \cite{Zhu3} and \cite{Zhu2}. From now on, we will abbreviate a multiple equal-difference representation to a MED representation.

\subsection{A reminder on MED representations of a cyclotomic coset}\label{sec 3.1}
\begin{definition}
	A subset $E \subseteq \mathbb{Z}/n\mathbb{Z}$ is called an equal-difference set if it has the form
	$$E = \{\gamma,\gamma+d,\cdots,\gamma+(\frac{n}{d}-1)d\},$$
	where $d$ is a positive divisor of $n$ and is called the common difference of $E$. In particular, an equal-difference $q$-cyclotomic coset modulo $n$ is a $q$-cyclotomic coset modulo $n$ that is an equal-difference set.
\end{definition}

For any $\gamma \in \mathbb{Z}/n\mathbb{Z}$ we set $n_{\gamma} = \frac{n}{\mathrm{gcd}(\gamma,n)}$. It is readily seen that for any $\gamma$ and $\gamma^{\prime}$ lying in the same $q$-cyclotomic coset modulo $n$ it holds that $n_{\gamma} = n_{\gamma^{\prime}}$. The criterion for a $q$-cyclotomic coset modulo $n$ being of equal difference is given below.

\begin{proposition}
	The cyclotomic coset $c_{n/q}(\gamma)$ is of equal difference if and only if the following two conditions are satisfied:
	\begin{description}
		\item[(\romannumeral1)] $\mathrm{rad}(n_{\gamma}) \mid q-1$;
		\item[(\romannumeral2)] $q \equiv 1 \pmod{4}$ if $8 \mid n_{\gamma}$.
	\end{description}
\end{proposition}

In general, a cyclotomic coset is not necessarily of equal difference, however it turns out that it can always be decomposed into a disjoint union of equal-difference subsets. We are particularly interested in those decomposition whose equal-difference components have the same common difference, which leads to the definition below.

\begin{definition}\label{def 1}
	Let $c_{n/q}(\gamma)$ be a $q$-cyclotomic coset modulo $n$. A multiple equal-difference representation (MED representation) of $c_{n/q}(\gamma)$ is a partition
	$$c_{n/q}(\gamma) = \bigsqcup_{i \in I}E_{i},$$
	where $E_{i}$, $i \in I$, are equal-difference sets with the same common difference. We denote by $\mathcal{MER}(c_{n/q}(\gamma))$ the set of all MED representations of $c_{n/q}(\gamma)$.
\end{definition}

There is a natural order on the set $\mathcal{MER}(c_{n/q}(\gamma))$. Let
$$c_{n/q}(\gamma) = \bigsqcup_{i \in I}E_{i} = \bigsqcup_{j \in J}E_{j}^{\prime}$$
be two MED representations of $c_{n/q}(\gamma)$. If the index set $J$ can be partitioned as $J = \bigsqcup\limits_{i \in I}J_{i}$, and for each $i \in I$ the equal-difference set $E_{i}$ can be further decomposed as
$$E_{i} = \bigsqcup_{j \in J_{i}}E_{j}^{\prime},$$
then we say that $\bigsqcup\limits_{i \in I}E_{i}$ is coarser than $\bigsqcup\limits_{j \in J}E_{j}^{\prime}$ (or $\bigsqcup\limits_{j \in J}E_{j}^{\prime}$ is finer than $\bigsqcup\limits_{i \in I}E_{i}$), and write $\bigsqcup\limits_{i \in I}E_{i} \geq \bigsqcup\limits_{j \in J}E_{j}^{\prime}$.

As any one-element subset of $\mathbb{Z}/n\mathbb{Z}$ is of equal-difference, the coset $c_{n/q}(\gamma)$ always admits the trivial MED representation
$$c_{n/q}(\gamma) = \bigsqcup_{i=0}^{\tau-1}\{\gamma q^{i}\},$$
where $\tau = |c_{n/q}(\gamma)|$. Clearly it is the unique finest MED representation of $c_{n/q}(\gamma)$. On the other hand, there is also a unique coarsest MED representation of $c_{n/q}(\gamma)$, given by the Lemma below.

\begin{lemma}\label{lem 4}
	Let
	\begin{equation*}
		\omega_{\gamma} = \left\{
		\begin{array}{lcl}
			2\mathrm{ord}_{\mathrm{rad}(n_{\gamma})}(q), \quad \mathrm{if} \ q^{\mathrm{ord}_{\mathrm{rad}(n_{\gamma})}(q)} \equiv 3 \pmod{4} \ \mathrm{and} \ 8 \mid n_{\gamma};\\
			\mathrm{ord}_{\mathrm{rad}(n_{\gamma})}(q), \quad \mathrm{otherwise}.
		\end{array} \right.
	\end{equation*}
	Then 
	$$c_{n/q}(\gamma) = \bigsqcup_{i=0}^{\omega_{\gamma}-1}c_{n/q^{\omega_{\gamma}}}(\gamma q^{i})$$
	is the unique coarsest MED representation of $c_{n/q}(\gamma)$.
\end{lemma}

Define a set
$$\Sigma_{\gamma} = \{d = \frac{nt}{\tau}: \ \omega_{\gamma} \mid t \ \mathrm{and} \ t \mid \tau\},$$
along with the order that for any $d_{1},d_{2} \in \Sigma_{\gamma}$, $d_{1} \leq d_{2}$ if and only if $d_{1} \mid d_{2}$. We are now ready to state the main theorem characterizing the structure of the space $\mathcal{MER}(c_{n/q}(\gamma))$ of MED representations of $c_{n/q}(\gamma)$. 

\begin{theorem}\label{thm 1}
	Let $c_{n/q}(\gamma)$ be a $q$-cyclotomic coset modulo $n$ with size $\tau$.  There is an anti-order-preserving one-to-one correspondence
	\begin{align*}
		\psi_{\gamma}: \Sigma_{\gamma} &\rightarrow \mathcal{MER}(c_{n/q}(\gamma));\\
		d &\mapsto \psi_{\gamma}(d): \ c_{n/q}(\gamma) = \bigsqcup_{i=0}^{t-1}c_{n/q^{t}}(\gamma q^{i}),
	\end{align*}
	where $t = \frac{d\tau}{n}$. Furthermore, for any $d \in \Sigma_{\gamma}$, the common difference of $\psi_{\gamma}(d)$ is $d$.
\end{theorem}

MED representations of $c_{n/q}(\gamma)$ are important because they completely determine the ways in which the induced polynomial 
 $M_{\gamma}(X)$ factors into binomials over extension fields of $\mathbb{F}_{q}$. Precisely it can be phrased as follows.

\begin{proposition}\label{prop 1}
	Let $f(X)$ be a nonconstant squarefree polynomial over $\mathbb{F}_{q}$ with a nonzero constant term. Let $n$ be the order of $f(X)$ and $T_{f} \subseteq \mathbb{Z}/n\mathbb{Z}$ be the defining set of $f(X)$. Then $f(X)$ is a binomial if and only if $T_{f}$ is of equal difference.
\end{proposition}

\begin{corollary}
	Let $c_{n/q}(\gamma)$ be a $q$-cyclotomic coset modulo $n$, and $M_{\gamma}(X)$ be the polynomial induced by $c_{n/q}(\gamma)$. Then $M_{\gamma}(X)$ is a binomial if and only if $c_{n/q}(\gamma)$ is an equal-difference cyclotomic coset.
\end{corollary}

On the other hand, any equal-difference subset
$$T = \{\gamma,\gamma+d,\cdots,\gamma+(\frac{n}{d}-1)d\} \subseteq \mathbb{Z}/n\mathbb{Z}$$
defines a binomial
\begin{align*}
	f_{T}(X) &= (X-\zeta_{n}^{\gamma})(X-\zeta_{n}^{\gamma+d})\cdots(X-\zeta_{n}^{\gamma+(\frac{n}{d}-1)d})\\
	&=\zeta_{n}^{\frac{\gamma n}{d}}(\frac{X}{\zeta_{n}^{\gamma}}-1)(\frac{X}{\zeta_{n}^{\gamma}}-\zeta_{\frac{n}{d}})\cdots(\frac{X}{\zeta_{n}^{\gamma}}-\zeta_{\frac{n}{d}}^{\frac{n}{d}-1})\\
	&=\zeta_{n}^{\frac{\gamma n}{d}}((\frac{X}{\zeta_{n}^{\gamma}})^{\frac{n}{d}}-1)\\
	&=X^{\frac{n}{d}}-\zeta_{d}^{\gamma},
\end{align*}
where the last equality holds since $\zeta_{n}^{\frac{n}{d}} = \zeta_{d}$. In general $f_{T}(X)$ is over the extension field $\mathbb{F}_{q}(\zeta_{n})$ of $\mathbb{F}_{q}$. The criterion for $f_{T}(X)$ being over $\mathbb{F}_{q}$ is given below.

\begin{lemma}\label{lem 1}
	The polynomial $f_{T}(X)$ is defined over $\mathbb{F}_{q}$ if and only if
	$$\gamma q \equiv \gamma \pmod{d}.$$
\end{lemma}

Consider a MED representation
\begin{equation}\label{eq 5}
	c_{n/q}(\gamma) = \bigsqcup_{i=0}^{t-1}c_{n/q^{t}}(\gamma q^{i}),
\end{equation}
of $c_{n/q}(\gamma)$, where $\omega_{\gamma} \mid t$ and $t \mid n$. It gives rise to a factorization
\begin{equation}\label{eq 7}
	M_{\gamma}(X) = \prod_{i=0}^{t-1}(X^{\frac{n}{d}}-\zeta_{d}^{\gamma q^{j}}),
\end{equation}
where $d = \frac{nt}{\tau}$. Note that \eqref{eq 7} is exactly the irreducible factorization of $M_{\gamma}(X)$ over $\mathbb{F}_{q^{t}}$. Thus by Proposition \ref{prop 1} we obtain the result below.

\begin{proposition}\label{thm 2}
    There is a one-to-one correspondence between $\mathcal{MER}(c_{n/q}(\gamma))$ and the set of extension fields of $\mathbb{F}_{q}$ contained in $\mathbb{F}_{q^{\tau}}$ over which $M_{\gamma}(X)$ factors into a product of irreducible binomials
    $$c_{n/q}(\gamma) = \bigsqcup_{i=0}^{t-1}c_{n/q^{t}}(\gamma q^{i}) \mapsto \mathbb{F}_{q^{t}},$$
    where $t$ is a divisor of $\tau$ that is divided by $\omega_{\gamma}$. Moreover, for any MED representations 
    $$c_{n/q}(\gamma) = \bigsqcup_{i=0}^{t_{1}-1}c_{n/q^{t_{1}}}(\gamma q^{i}) = \bigsqcup_{j=0}^{t_{2}-1}c_{n/q^{t_{2}}}(\gamma q^{j}),$$
    the representations satisfy $\bigsqcup\limits_{j=0}^{t_{2}-1}c_{n/q^{t_{2}}}(\gamma q^{j}) \leq \bigsqcup\limits_{i=0}^{t_{1}-1}c_{n/q^{t_{1}}}(\gamma q^{i})$ if and only if $\mathbb{F}_{q^{t_{1}}} \subseteq \mathbb{F}_{q^{t_{2}}}$.
\end{proposition}

As $\mathbb{F}_{q^{\tau}}$ is the splitting field of $M_{\gamma}(X)$ over $\mathbb{F}_{q}$, for any positive integer $t$ the irreducible factorization of $M_{\gamma}(X)$ over $\mathbb{F}_{q^{t}}$ is the same as that over $\mathbb{F}_{q^{t^{\prime}}}$ for $t^{\prime} = \mathrm{gcd}(t,\tau)$. The following corollary is an immediate consequence of Theorem \ref{thm 1} and Proposition \ref{thm 2}.

\begin{corollary}\label{coro 1}
	The family
	$$\{\mathbb{F}_{q^{t}}: \ \frac{\mathrm{gcd}(t,\tau)n}{\tau} \in \Sigma_{\gamma}\}$$
	of fields are exactly the extensions of $\mathbb{F}_{q}$ over which $M_{\gamma}(X)$ factors into a product of irreducible binomials.
\end{corollary}

\subsection{MED representations of the defining set of a squarefree polynomial}\label{sec 3.2}
Let $f(X)$ be a nonconstant simple-root polynomial over $\mathbb{F}_{q}$, with a nonzero constant term. Let $n$ be the order of $f(X)$, and $T_{f} \subseteq \mathbb{Z}/n\mathbb{Z}$ be the defining set of $f(X)$. Generalizing Definition \ref{def 1}, we define a MED representation of $T_{f}$.

\begin{definition}
	A multiple equal-difference representation (MED representation) of $T_{f}$ is a partition
	\begin{equation}\label{eq 1}
		T_{f} = \bigsqcup_{i \in I}E_{i},
	\end{equation}
	where $E_{i}$, $i \in I$, are equal-difference sets with the same common difference $d$. The integer $d$ is called the common difference of the MED representation \eqref{eq 1}. We denote by $\mathcal{MER}(T_{f})$ the space of all MED representations of $T_{f}$.
\end{definition}

The order on $\mathcal{MER}(T_{f})$ can be defined similarly as in the case of a cyclotomic coset. For two MED representations
$$T_{f} = \bigsqcup_{i \in I}E_{i} = \bigsqcup_{j \in J}E_{j}^{\prime}$$
of $T_{f}$. If the index set $J$ can be partitioned as $J = \bigsqcup\limits_{i \in I}J_{i}$, and for each $i \in I$ the equal-difference set $E_{i}$ can be further decomposed as
$$E_{i} = \bigsqcup_{j \in J_{i}}E_{j}^{\prime},$$
then we say that $\bigsqcup\limits_{i \in I}E_{i}$ is coarser than $\bigsqcup\limits_{j \in J}E_{j}^{\prime}$ (or $\bigsqcup\limits_{j \in J}E_{j}^{\prime}$ is finer than $\bigsqcup\limits_{i \in I}E_{i}$), and write $\bigsqcup\limits_{i \in I}E_{i} \geq \bigsqcup\limits_{j \in J}E_{j}^{\prime}$.

We present two concrete examples of a MED representation of $T_{f}$, which shows that $T_{f}$ always admits a MED representation.

\begin{example}
	Since any one-element subset of $\mathbb{Z}/n\mathbb{Z}$ is an equal-difference subset, the defining set $T_{f}$ has the trivial MED representation
	$$T_{f} = \bigsqcup_{\gamma \in T_{f}}\{\gamma\}.$$
	Clearly it is the finest MED representation of $T_{f}$.
\end{example}

\begin{example}
	As $f(X)$ can be factorized into a product of irreducible factors of $X^{n}-1$, $T_{f}$ can be represented as a disjoint union of $q$-cyclotomic cosets modulo $n$, say,
	$$T_{f} = \bigsqcup_{j=1}^{r}c_{n/q}(\gamma_{j}).$$
	For $1 \leq j \leq r$ let $\tau_{j}$ be the size of $c_{n/q}(\gamma_{j})$, and let
	\begin{equation*}
		\omega_{\gamma_{j}} = \left\{
		\begin{array}{lcl}
			2\mathrm{ord}_{\mathrm{rad}(n_{\gamma_{j}})}(q), \quad \mathrm{if} \ q^{\mathrm{ord}_{\mathrm{rad}(n_{\gamma_{j}})}(q)} \equiv 3 \pmod{4} \ \mathrm{and} \ 8 \mid n_{\gamma_{j}};\\
			\mathrm{ord}_{\mathrm{rad}(n_{\gamma_{j}})}(q), \quad \mathrm{otherwise}.
		\end{array} \right.
	\end{equation*}
	Set $u = \mathrm{gcd}(\frac{\tau_{1}}{\omega_{1}},\cdots,\frac{\tau_{r}}{\omega_{r}})$ and
	$$\omega_{j}^{\prime} = \frac{\tau_{j}}{u}, \ 1\leq j\leq r.$$
	For every $1 \leq j \leq r$ one has $\omega_{j} \mid \omega_{j}^{\prime}$ and $\omega_{j}^{\prime} \mid \tau_{j}$, therefore, by Theorem \ref{thm 1} 
	$$c_{n/q}(\gamma_{j}) = \bigsqcup_{i_{j}=0}^{\omega_{j}^{\prime}-1}c_{n/q^{\omega_{j}^{\prime}}}(\gamma_{i}q^{i_{j}})$$
	is a MED representation of $c_{n/q}(\gamma_{j})$, with common difference  
	$$d = \frac{n}{|c_{n/q^{\omega_{j}^{\prime}}}(\gamma_{j}q^{i_{j}})|} = \frac{n\omega_{j}^{\prime}}{\tau_{j}} = \frac{n}{u}.$$
	It follows that 
	$$T_{f} = \bigsqcup_{j=1}^{r}\bigsqcup_{i_{j}=0}^{\omega_{j}^{\prime}-1}c_{n/q^{\omega_{j}^{\prime}}}(\gamma_{j}q^{i_{j}})$$
	is a MED representation of $T_{f}$ with common difference $d = \frac{n}{u}$.
\end{example}

Let $a \in \mathbb{Z}/n\mathbb{Z}$. Define a translation $t_{a}$ by
$$t_{a}: \mathbb{Z}/n\mathbb{Z} \rightarrow \mathbb{Z}/n\mathbb{Z}; \ \gamma \mapsto \gamma+a.$$
If $T_{f} \subseteq \mathbb{Z}/n\mathbb{Z}$ is stable under $t_{a}$, that is, $t_{a}(T_{f}) = T_{f}$, then $a$ is called a period of $T_{f}$. It is straightforward to verify that the set $\overline{\Sigma}_{f}$ of periods of $T_{f}$ forms an additive subgroup of $\mathbb{Z}/n\mathbb{Z}$, which is referred to as the additive stabilizer of $T_{f}$.

\begin{lemma}\label{lem 5}
	\begin{description}
		\item[(1)] If $a,b \in \overline{\Sigma}_{f}$, then $\mathrm{gcd}(a,b) \in \overline{\Sigma}_{f}$.
		\item[(2)] If $a \in \overline{\Sigma}_{f}$, then $\mathrm{gcd}(a,n) \in \overline{\Sigma}_{f}$.
	\end{description}
\end{lemma}

\begin{proof}
	There are integers $u,v$ such that
	$$\mathrm{gcd}(a,b) = ua+vb.$$
	Then $(1)$ follows from that $\overline{\Sigma}_{f}$ is a group. And note that $n=0 \in \overline{\Sigma}_{f}$, then $(2)$ is an immediate consequence of $(1)$.
\end{proof}

\begin{definition}
	Define the common difference set associated to $T_{f}$ to be
	$$\Sigma_{f} = \{1 \leq d \leq n: \ d \mid n \ \mathrm{and} \ \overline{d} \in \overline{\Sigma}_{f}\},$$
	where $\overline{d}$ denotes the residue class of $d$ modulo $n$.
\end{definition}

Unless stating otherwise, we write elements in $\overline{\Sigma}_{f}$ as integers $a$ in the range $1 \leq a \leq n$. Then the common difference set $\Sigma_{f}$ can be identified with the subset of $\overline{\Sigma}_{f}$ consisting of the integers dividing $n$.

On the other hand, the additive stabilizer $\overline{\Sigma}_{f}$ can be recovered from $\Sigma_{f}$ as follows. Let $d_{0}$ be the smallest integer in $\Sigma_{f}$. By $(2)$ of Lemma \ref{lem 5} it is also the smallest integer in $\overline{\Sigma}_{f}$. As $\overline{\Sigma}_{f}$ is a subgroup of $\mathbb{Z}/n\mathbb{Z}$, we have
$$\overline{\Sigma}_{f} = d_{0}\mathbb{Z}/n\mathbb{Z}.$$

Define an order $\leq$ on $\Sigma_{f}$ by setting, for $d_{1},d_{2} \in \Sigma_{f}$, $d_{1} \leq d_{2}$ if $d_{1} \mid d_{2}$. Then $\Sigma_{f}$ has a unique maximal element and a unique minimal element, which are given by the lemma below.

\begin{lemma}\label{lem 9}
	Let the notations be given as above. Then we have that
	\begin{description}
		\item[(1)] $n$ is the unique maximal element in $\Sigma_{f}$, and
		\item[(2)] $d_{0}$ is the unique minimal element in $\Sigma_{f}$.
	\end{description}
\end{lemma}

\begin{proof}
	The assertion $(1)$ is clear from the definition of $\Sigma_{f}$. By the argument before Lemma \ref{lem 9}
	$$\overline{\Sigma}_{f} = d_{0}\mathbb{Z}/n\mathbb{Z},$$
	so viewing elements in $\overline{\Sigma}_{f}$ as integers in the range $1 \leq a \leq n$, every element in $\overline{\Sigma}_{f}$ is a multiple of $d_{0}$. Then the conclusion follows from the embedding of $\Sigma_{f}$ into $\overline{\Sigma}_{f}$.
\end{proof}

Now we can prove the main theorem in this section, which claims that $\Sigma_{f}$ gives an anti-order-preserving complete classification of the space $\mathcal{MER}(T_{f})$ of MED representations of $T_{f}$.

\begin{theorem}\label{thm 3}
	There is an anti-order-preserving one-to-one correspondence between $\Sigma_{f}$ and $\mathcal{MER}(T_{f})$.
\end{theorem}

\begin{proof}
	To any $d \in \Sigma_{f}$, we assign a MED representation of $T_{f}$ as follows. Choose a $\gamma_{1} \in T_{f}$. As $t_{d}(T_{f}) = T_{f}$, then
	$$E_{1} = \{\gamma_{1}+jd \ | \ 0 \leq j \leq \frac{n}{d}-1\} \subseteq T_{f}.$$
	If $E_{1} = T_{f}$, there is nothing further to show. Otherwise, there is a $\gamma_{2} \in T_{f} \setminus E_{1}$. Similarly we have
	$$E_{2} =  \{\gamma_{2}+jd \ | \ 0 \leq j \leq \frac{n}{d}-1\} \subseteq T_{f}.$$
	As $\gamma_{2} \notin E_{1}$, then $E_{1} \cap E_{2} = \varnothing$ and consequently 
	$$E_{1} \sqcup E_{2} \subseteq T_{f}.$$
Keep on the above process. Since $T_{f}$ is a finite set, it will stop after finitely many steps and yield a 
MED representation
	\begin{equation}\label{eq 16}
		T_{f} = \bigsqcup_{i=1}^{s}E_{i}
	\end{equation}
	of $T_{f}$, with common difference $d$. Denoting the MED representation \eqref{eq 16} by $\psi_{f}(d)$, then we obtain a map
	$$\psi_{f}: \Sigma_{f} \rightarrow \mathcal{MER}(T_{f}).$$
	
	Notice that for any $d \in \Sigma_{f}$ the common difference of $\psi_{f}(d)$ is $d$, thus the map $\psi_{f}$ is injective. To show that $\psi_{f}$ is surjective, it suffices to prove the following two assertions:
	\begin{itemize}
		\item[(\romannumeral1)] the common difference of any MED representation of $T_{f}$ lies in $\Sigma_{f}$;
		\item[(\romannumeral2)] for any $d \in \Sigma_{f}$, $\psi_{f}(d)$ is the only MED representation of $T_{f}$ with common difference $d$.
	\end{itemize}
	Let 
	$$T_{f} = \bigsqcup_{i=1}^{r}F_{i}$$
	be a MED representation of $T_{f}$, with common difference $d$. As $t_{d}(F_{i}) = F_{i}$ for any $1 \leq i \leq r$, then
	$$t_{d}(T_{f}) = \bigsqcup_{i=1}^{r}t_{d}(F_{i}) = \bigsqcup_{i=1}^{r}F_{i} = T_{f},$$
	which implies that $d \in \Sigma_{f}$. This proves the assertion (\romannumeral1). To prove (\romannumeral2), we write $\psi_{f}(d)$ as
	$$T_{f} = \bigsqcup_{i=1}^{s}E_{i}.$$
	Choose a $\gamma_{1} \in E_{1}$. Since $\gamma_{1} \in T_{f} = \bigsqcup\limits_{i=1}^{r}F_{i}$, changing the order the $F_{i}$'s if necessary, then $\gamma_{1} \in F_{1}$. Notice that any two equal-difference sets with common difference $d$ either coincide or disjoint, hence we have $E_{1} = F_{1}$ and 
	$$\bigsqcup_{i=2}^{r}F_{i} = \bigsqcup_{i=2}^{s}E_{i}.$$
	Applying induction yields that $r=s$ and $E_{i} = F_{i}$ for all $1 \leq i \leq s$ after changing the order of the $F_{i}$'s suitably.
	
	Finally it remains to prove that the map $\psi$ is anti-order-preserving. Let $d,d^{\prime} \in \Sigma_{f}$ be such that $d \mid d^{\prime}$. Write $\psi_{f}(d)$ as
	$$T_{f} = \bigsqcup_{i=1}^{s}E_{i}.$$
	For each $1 \leq i \leq s$ the equal-difference set $E_{i}$ can be expressed as
	$$E_{i} = \gamma_{i} + d\cdot \mathbb{Z}/n\mathbb{Z}$$
	for some $\gamma_{i} \in \mathbb{Z}/n\mathbb{Z}$. Then $E_{i}$ can be further partitioned into 
	$$E_{i} = \bigsqcup_{j_{i}=0}^{\frac{d^{\prime}}{d}-1}(\gamma_{i}+j_{i}d+d^{\prime}\cdot\mathbb{Z}/n\mathbb{Z}),$$
	which gives
	\begin{equation}\label{eq 17}
		T_{f} = \bigsqcup_{i=1}^{s}\bigsqcup_{j_{i}=0}^{\frac{d^{\prime}}{d}-1}(\gamma_{i}+j_{i}d+d^{\prime}\cdot\mathbb{Z}/n\mathbb{Z}).
	\end{equation}
	Note that for any $1\leq i \leq s$ and any $0 \leq j_{i} \leq \frac{d^{\prime}}{d}$, the set $\gamma_{i}+j_{i}d+d^{\prime}\cdot\mathbb{Z}/n\mathbb{Z}$ is an equal-difference set with common difference $d^{\prime}$. Thus \eqref{eq 17} is a MED representation of $T_{f}$ with common difference $d^{\prime}$. By the assertion (\romannumeral2) in the last paragraph, \eqref{eq 17} coincides with $\psi_{f}(d^{\prime})$. Hence one has $\psi_{f}(d^{\prime}) \leq \psi_{f}(d)$.
\end{proof}

\begin{remark}
	From the proof of Theorem \ref{thm 3}, $\Sigma_{f}$ can be realized as the set of the integers that appear as the common differences of MED representations of $T_{f}$, which justifies its name. In particular, in the case where $f(X)$ is irreducible over $\mathbb{F}_{q}$, Theorem \ref{thm 1} can be rephrased as below.
\end{remark}

\begin{corollary}
	If $f(X)$ is an irreducible polynomial over $\mathbb{F}_{q}$, with defining set $c_{n/q}(\gamma)$, then the common difference set associated to $c_{n/q}(\gamma)$ is $\Sigma_{\gamma}$, and the additive stabilizer of $c_{n/q}(\gamma)$ is $\omega_{\gamma}\mathbb{Z}/n\mathbb{Z}$.
\end{corollary}

Combining Lemma \ref{lem 9} and Theorem \ref{thm 3}, we obtain the following consequences immediately.

\begin{corollary}\label{coro 3}
	There is a unique coarsest MED representation of $T_{f}$.
\end{corollary}

\begin{corollary}
	The defining set $T_{f}$ admits a nontrivial MED representation if and only if $\overline{\Sigma}_{f}$ is not a trivial group.
\end{corollary}

For a nonconstant squarefree polynomial $f(X) \in \mathbb{F}_{q}[X]$ such that $f(0) \neq 0$, its defining set $T_{f}$ depends on the choice of a primitive $n$-th root of unity, where $n= \mathrm{ord}(f)$. However, it turns out that the additive stabilizer $\overline{\Sigma}_{f}$ and the common difference set $\Sigma_{f}$ associated to $T_{f}$ are independent of this choice.

Assume that $\zeta_{n}$ and $\zeta_{n}^{\prime}$ are two primitive $n$-th roots of unity, and $T_{f}$ and $T_{f}^{\prime}$ are the defining sets of $f(X)$ with respect to $\zeta_{n}$ and $\zeta_{n}^{\prime}$, respectively. There is a positive integer $r$ coprime to $n$ such that $\zeta_{n}^{\prime} = \zeta_{n}^{r}$, which gives
$$T_{f}^{\prime} = r^{-1}T_{f} = \{r^{-1}\gamma \ | \ \gamma \in T_{f}\}.$$
Given a MED representation
$$T_{f} = \bigsqcup_{i \in I}E_{i}$$
of $T_{f}$ with common difference $d$, we have
$$T_{f}^{\prime} = r^{-1}T_{f} = \bigsqcup_{i \in I}r^{-1}E_{i}.$$
For every $i \in I$, $E_{i}$ can be written as
$$E_{i} = \{\gamma_{i}+jd \ | \ 0 \leq j \leq \frac{n}{d}-1\}$$
for some $\gamma_{i} \in \mathbb{Z}/n\mathbb{Z}$. Noting that $r^{-1}$ is coprime to $n$, one obtains that
$$r^{-1}E_{i} = \{r^{-1}\gamma_{i}+r^{-1}jd \ | \ 0 \leq j \leq \frac{n}{d}-1\} = \{r^{-1}\gamma_{i}+jd \ | \ 0 \leq j \leq \frac{n}{d}-1\},$$
which is an equal-difference set with common difference $d$. It follows that 
$$T_{f}^{\prime} = \bigsqcup_{i \in I}r^{-1}E_{i}$$
is a MED representation of $T_{f}^{\prime}$ with common difference $d$. 

Note that the common difference set $\Sigma_{f}$ (resp. $\Sigma_{f}^{\prime}$) associated to $T_{f}$ (resp. $T_{f}^{\prime}$) consists of the common differences of MED representations of $T_{f}$ (resp. $T_{f}^{\prime}$), thus the above argument shows that they coincide. Moreover, as the additive stabilizer $\overline{\Sigma}_{f}$ (resp. $\overline{\Sigma}_{f}^{\prime}$) of $T_{f}$ (resp. $T_{f}^{\prime}$) can be written as
$$\overline{\Sigma}_{f} = d_{0}\mathbb{Z}/n\mathbb{Z} = \overline{\Sigma}_{f}^{\prime},$$
where $d_{0}$ is the unique minimal element in $\Sigma_{f} = \Sigma_{f}^{\prime}$. This proves that the additive stabilizer and the common difference set do not depend on the particular form of the defining set $T_{f}$, which explains why we denote them respectively by $\overline{\Sigma}_{f}$ and $\Sigma_{f}$ instead of $\overline{\Sigma}_{T_{f}}$ and $\Sigma_{T_{f}}$.

\subsection{The arithmetic structure of $\Sigma_{f}$}
We adopt the conventions and notations from the last subsection. Proposition \ref{prop 1} and \ref{thm 3} together lead to the following result.

\begin{proposition}\label{prop 3}
	There is a one-to-one correspondence between $\Sigma_{f}$ and the set of all factorizations of $f(X)$ into a product of binomials of the same degree over $\mathbb{F}_{q}(\zeta_{n})$, where $n = \mathrm{ord}(f)$. Moreover, this correspondence is anti-order-preserving in the sense that for any $d_{1}, d_{2} \in \Sigma_{f}$, $d_{1} \mid d_{2}$ if and only if the factorization corresponding to $d_{1}$ can be obtained by further decomposing the factors that appear in the factorization corresponding to $d_{2}$.
\end{proposition}

\begin{proof}
	Let 
	$$T_{f} = \bigsqcup_{i=1}^{s}E_{i}$$
	be a MED representation of $T_{f}$ with common difference $d$. Set
	$$f_{i}(X) = \prod_{\gamma \in E_{i}}(X - \zeta_{n}^{\gamma})$$
	for $1 \leq i \leq s$. Then
	$$f(X) = \prod_{i=1}^{s}f_{i}(X)$$
	is a factorization of $f(X)$ into a product of binomials of degree $\frac{\mathrm{deg}(f)}{d}$.
	
	Conversely, if
	\begin{equation}\label{eq 18}
		f(X) = \prod_{i=1}^{s}f_{i}(X)
	\end{equation}
	is a factorization of $f(X)$ into a product of binomials of degree $r$, by Theorem \ref{thm 1} the defining set $E_{i}$ of $f_{i}(X)$ is an equal-difference set of common difference $\frac{\mathrm{deg}(f)}{r}$. Therefore \eqref{eq 18} induces a MED representation
	$$T_{f} = \bigsqcup_{i=1}^{s}E_{i}.$$
	
	Let 
	$$T_{f} = \bigsqcup_{i \in I}E_{i} = \bigsqcup_{j \in J}E_{j}^{\prime}$$
	be two MED representations of $T_{f}$, and
	$$f(X) = \prod_{i \in I}f_{i}(X) = \prod_{j \in J}f_{j}^{\prime}(X)$$
	be the corresponding factorizations of $f(X)$. By definition $\bigsqcup\limits_{i \in I}E_{i} \geq \bigsqcup\limits_{j \in J}E_{j}^{\prime}$ if and only if the index set $J$ can be partitioned as $J = \bigsqcup\limits_{i \in I}J_{i}$, and for each $i \in I$ the equal-difference set $E_{i}$ can be further decomposed as
	\begin{equation}\label{eq 8}
		E_{i} = \bigsqcup_{j \in J_{i}}E_{j}^{\prime}.
	\end{equation}
	The latter condition \eqref{eq 8} is equivalent to that each $f_{i}(X)$ can be decomposed as
	$$f_{i}(X) = \prod_{j \in J_{i}}f_{j}^{\prime}(X).$$
	
	Now Proposition \ref{prop 3} follows from Theorem \ref{thm 2}.
\end{proof}

Proposition \ref{prop 3} says that the common difference set $\Sigma_{f}$ of $T_{f}$ completely classifies the factorizations of $f(X)$ into a product of binomials of the same degree. However, in general these factorizations are over the splitting field of $f(X)$. It is also interesting to determine, given any finite extension field $\mathbb{F}_{q^{t}}$ of $\mathbb{F}_{q}$, the factorizations that are over $\mathbb{F}_{q^{t}}$.

Denote by $\mathbb{F}_{q^{\theta}}$ the splitting field of $f(X)$ over $\mathbb{F}_{q}$. We first show that for any $t \in \mathbb{N}^{+}$, a factor $g(X)$ is over $\mathbb{F}_{q^{t}}$ if and only if it is over $\mathbb{F}_{q^{t^{\prime}}}$ for $t^{\prime} = \mathrm{gcd}(t,\theta)$. If $g(X)$ is over $\mathbb{F}_{q^{t^{\prime}}}$, it is certainly over $\mathbb{F}_{q^{t}}$. Conversely, if $g(X)$ is over $\mathbb{F}_{q^{t}}$, then the defining set $T_{g}$ of $g(X)$ satisfies $q^{t}T_{g} = T_{g}$. Since $\mathbb{F}_{q^{\theta}}$ is the splitting field of $f(X)$, it clearly holds that $q^{\theta}T_{g} = T_{g}$. Writing $t^{\prime}$ as $t^{\prime} = tu+\theta v$, then
$$q^{t^{\prime}}T_{g} = q^{tu}(q^{\theta v}T_{g}) = q^{tu}T_{g} = T_{g},$$
which implies that $g(X)$ is over $\mathbb{F}_{q^{t^{\prime}}}$. Further, this conclusion indicates that the factorizations of $f(X)$ over $\mathbb{F}_{q^{t}}$ are the same as those over $\mathbb{F}_{q^{t^{\prime}}}$.

\begin{theorem}\label{thm 6}
	For any $t \in \mathbb{N}^{+}$ we set $t^{\prime} = \mathrm{gcd}(t,\theta)$, and define
	$$\Sigma_{f,t^{\prime}} = \{d \in \Sigma_{f} \ | \ q^{t^{\prime}} \equiv 1 \pmod{d}\}.$$
	Then under the bijection given in Proposition \ref{prop 2}, $\Sigma_{f,t^{\prime}}$ corresponds exactly to the factorizations of $f(X)$ over $\mathbb{F}_{q^{t}}$ into a product of binomials of the same degree.
\end{theorem}

\begin{proof}
	Let $d \in \Sigma_{f}$, and let
	\begin{equation}\label{eq 9}
		T_{f} = \bigsqcup_{i=1}^{s}E_{i}
	\end{equation}
	be the MED representation of $T_{f}$ with common difference $d$. For $1 \leq i \leq s$ we write
	$$E_{i} = \{\gamma_{i}+jd \ | \ 0 \leq j \leq \frac{n}{d}\}.$$
	Then the factorization of $f(X)$ induced by \eqref{eq 9} is given by
	\begin{equation}\label{eq 10}
		f(X) = \prod_{i=1}^{s}f_{i}(X),
	\end{equation}
	where for each $i$
	$$f_{i}(X) = \prod_{j=0}^{\frac{n}{d}-1}(X-\zeta_{n}^{\gamma_{i}+jd}) = X^{\frac{n}{d}}-\zeta_{d}^{\gamma_{i}}.$$
	The factorization \eqref{eq 10} is over $\mathbb{F}_{q^{t}}$ if and only if it is over $\mathbb{F}_{q^{t^{\prime}}}$, which by Lemma \ref{lem 1} amounts to that for all $1 \leq i \leq s$
	\begin{equation}\label{eq 11}
		\gamma_{i}q^{t^{\prime}} \equiv \gamma_{i} \pmod{d}.
	\end{equation}
	Notice that $n$ is the order of $f(X)$, therefore the greatest common divisor of the $\gamma_{i}$'s is coprime to $n$, and hence is coprime to $d$. It then follows that \eqref{eq 11} holds for every $1 \leq i \leq s$ if and only if
	$$q^{t^{\prime}} \equiv 1 \pmod{d}.$$
\end{proof}

Following from Theorem \ref{thm 4}, for any positive divisor $t$ of $\theta$ we call the set $\Sigma_{f,t}$ the common difference subset of $T_{f}$ over $\mathbb{F}_{q^{t}}$. Applying these notations, Corollary \ref{coro 1} can be rephrased as follows.

\begin{corollary}
	If $f(X)$ is irreducible over $\mathbb{F}_{q}$, with defining set $c_{n/q}(\gamma)$, then for any positive divisor $t$ of $\tau = |c_{n/q}(\gamma)|$ the common difference subset of $c_{n/q}(\gamma)$ is 
	$$\Sigma_{f,t} = \{d\in \Sigma_{\gamma}: \ d \mid \frac{nt}{\tau}\}.$$
\end{corollary}

\section{The arithmetic Singleton bound on the Hamming distances of simple-root constacyclic codes}\label{sec 4}
\subsection{The definition of arithmetic Singleton bound and basic properties}\label{sec 4.1}
In this section, we assume that $m$ is a positive integer coprime to $q$, and $\lambda$ is a nonzero element in $\mathbb{F}_{q}$. Let
$$\mathcal{C} = (f(X)) \subseteq \mathbb{F}_{q}[X]/(X^{m}-\lambda)$$
be a $\lambda$-constacyclic code of length $m$, with the generator polynomial $f(X) \mid X^{m}-\lambda$. Denote by $\tau$ the degree of $f(X)$, by $n$ the order of $f(X)$, and by $T_{f} \subseteq \mathbb{Z}/n\mathbb{Z}$ the defining set of $f(X)$.

\begin{proposition}\label{prop 2}
	Let
	$$T_{f} = \bigsqcup_{i=1}^{s}E_{i}$$
	be a MED representation of $T_{f}$ with common difference $d$. Then the Hamming distance $d_{H}(\mathcal{C})$ of $\mathcal{C}$ satisfies
	$$d_{H}(\mathcal{C}) \leq \frac{\tau d}{n}+1.$$
\end{proposition}

\begin{proof}
	First since $|T_{f}|=\tau$ and $|E_{i}|=\frac{n}{d}$ for any $1 \leq i \leq s$, then $s = \frac{\tau d}{n}$. Write each $E_{i}$ as
	$$E_{i} = \{\gamma_{i}+jd \ | \ 0 \leq j \leq \frac{n}{d}-1\}$$
	for some $\gamma_{i} \in \mathbb{Z}/n\mathbb{Z}$. Applying the formula before \ref{lem 1}, the polynomial $f(X)$ can be factorized as
	$$f(X) = \prod_{i=1}^{s}(X^{\frac{n}{d}}-\zeta_{n}^{\frac{\gamma_{i}n}{d}})$$
	over $\mathbb{F}_{q}(\zeta_{n})$. Setting
	$$f_{0}(X) = \prod_{i=1}^{s}(X-\zeta_{n}^{\frac{\gamma_{i}n}{d}}),$$
	then clearly $f(X) = f_{0}(X^{\frac{n}{d}})$. The substitution $X\mapsto X^{n/d}$ 
preserves the number of nonzero coefficients. It follows that $f(X)$ and $f_{0}(X)$ have the same Hamming weight. Note that the degree of $f_{0}(X)$ is $s$, therefore
	$$\mathrm{wt}_{H}(f) = \mathrm{wt}_{H}(f_{0}) \leq s+1 = \frac{\tau d}{n}+1.$$
	Finally, as $f(X)$ lies in $\mathcal{C} = (f(X))$, then
	$$d_{H}(\mathcal{C}) \leq \frac{\tau d}{n}+1.$$
\end{proof}

\begin{corollary}\label{coro 4}
	The trivial MED representation of $T_{f}$ gives rise to the Singleton bound.
\end{corollary}

\begin{proof}
	The trivial MED representation
	$$T_{f} = \bigsqcup_{\gamma \in T_{f}}\{\gamma\}$$
	has common difference $n$, hence by Proposition \ref{prop 2} it gives the upper bound
	$$d_{H}(\mathcal{C}) \leq \tau+1 = \mathrm{deg}(f)+1$$
	for the Hamming distance of $\mathcal{C}$, which is exactly the Singleton bound.
\end{proof}

Proposition \ref{prop 2} yields a family
$$\Gamma_{f} = \{\frac{\tau d}{n} \ | \ d \in \Sigma_{f}\},$$
of upper bounds on the Hamming distance of $\mathcal{C}$ arising from the MED representations of $T_{f}$. Corollary \ref{coro 4} indicates that the weakest bound lying in $\Gamma_{f}$ is exactly the classical Singleton bound. This can be interpreted as saying that the structure of an ideal in a quotient of the polynomial ring on a constacyclic code imposes stronger restriction on the Hamming distance, which can be partially revealed by the family $\Gamma_{f}$ of bounds. In this sense, we introduce the following definition.

\begin{definition}
	The smallest positive integer lying in $\Gamma_{f}$ is called the arithmetic Singleton bound for the Hamming distance of $\mathcal{C}$.
\end{definition}

Based on the results in the last section, the arithmetic Singleton bound can be equivalently characterized as below.

\begin{theorem}\label{thm 4}
	Let $\overline{\Sigma}_{f}$ be defined as in Subsection \ref{sec 3.2}, and let $\sigma_{f} = |\overline{\Sigma}_{f}|$. Then the arithmetic Singleton bound on the Hamming distance of $\mathcal{C} = (f(X))$ is equal to $\frac{\tau}{\sigma_{f}}+1$. 
Hence the arithmetic Singleton bound is the strongest upper bound obtainable from MED representations.
\end{theorem}

\begin{proof}
	By definition the arithmetic Singleton bound is $\frac{\tau d_{0}}{n}+1$, where $d_{0}$ is the smallest positive integer in $\Sigma_{f}$. On the other hand, the group $\overline{\Sigma}_{f}$ can be written as
	$$\overline{\Sigma}_{f} = d_{0}\mathbb{Z}/n\mathbb{Z},$$
	which implies that $\sigma_{f} = |\overline{\Sigma}_{f}| = \frac{n}{d_{0}}$. It follows that the arithmetic Singleton bound is given by
	$$\frac{\tau d_{0}}{n}+1 = \frac{\tau}{n}\cdot \frac{n}{\sigma_{f}}+1 = \frac{\tau}{\sigma_{f}}+1.$$
\end{proof}

\begin{corollary}\label{coro 7}
	Denote by $b_{\mathrm{S}}(\mathcal{C})$ and $b_{\mathrm{AS}}(\mathcal{C})$ the Singleton bound and the arithmetic Singleton bound for $\mathcal{C}$ respectively. Then the three statements below are equivalent:
	\begin{description}
		\item[(\romannumeral1)] $b_{\mathrm{S}}(\mathcal{C}) = b_{\mathrm{AS}}(\mathcal{C})$;
		\item[(\romannumeral2)] the group $\overline{\Sigma}_{f}$ is trivial;
		\item[(\romannumeral3)] the defining set $T_{f}$ only admits the trivial MED representation.
	\end{description}
\end{corollary}

\begin{proof}
	\textbf{(\romannumeral1)} $\Rightarrow$ \textbf{(\romannumeral2)}: By Theorem \ref{thm 4} the arithmetic Singleton bound is $b_{\mathrm{S}}(\mathcal{C}) = \frac{\tau}{\sigma_{f}}+1$. Then $b_{\mathrm{S}}(\mathcal{C}) = b_{\mathrm{AS}}(\mathcal{C})$ amounts to $\sigma_{f}=1$, which implies that the group $\overline{\Sigma}_{f}$ is trivial.
	
	\textbf{(\romannumeral2)} $\Rightarrow$ \textbf{(\romannumeral3)}: Since $\Sigma_{f}$ is a nonempty subset of $\overline{\Sigma}_{f}$, then one has
	$$\Sigma_{f} = \overline{\Sigma}_{f} = \{n\}.$$
	
	\textbf{(\romannumeral3)} $\Rightarrow$ \textbf{(\romannumeral1)}: This is an immediate consequence of Corollary \ref{coro 4} and the definition of arithmetic Singleton bound.
\end{proof}

\subsection{The arithmetic Singleton bound for irreducible constacyclic codes}\label{sec 4.2}
Let $m$ be a positive integer, and $\lambda$ be a nonzero element in $\mathbb{F}_{q}$. Let 
$$\mathcal{C} = (f(X)) \subseteq \mathbb{F}_{q}[X]/(X^{m}-\lambda)$$
be a $\lambda$-constacyclic code of length $m$ over $\mathbb{F}_{q}$, with the generator polynomial $f(X)$. We say that $\mathcal{C}$ is an irreducible constacyclic code if $f(X)$ is irreducible. In this subsection we concentrate on this case. We derive an explicit formula for the arithmetic Singleton bound of $\mathcal{C}$, and compare it with the classical Singleton bound. In particular, we show that the arithmetic Singleton bound for $\mathcal{C}$ is completely determined by $q$ and the order of $f(X)$.

Assume that $n$ is the order of $f(X)$, with the factorization into prime powers given by
$$n = p_{1}^{e_{1}}\cdots p_{r}^{e_{r}},$$
where $p_{1},\cdots,p_{r}$ are distinct prime numbers different from $p$ and $e_{1},\cdots,e_{r}$ are positive integers. Set 
\begin{equation*}
	\omega = \left\{
	\begin{array}{lcl}
		2\mathrm{ord}_{p_{1}\cdots p_{r}}(q), \quad \mathrm{if} \ q^{\mathrm{ord}_{p_{1}\cdots p_{r}}(q)} \equiv 3 \pmod{4} \ \mathrm{and} \ 8 \mid n;\\
		\mathrm{ord}_{p_{1}\cdots p_{r}}(q), \quad \mathrm{otherwise}.
	\end{array} \right.
\end{equation*} 
The Singleton bound and the arithmetic Singleton bound are given concretely by the following theorem.

\begin{theorem}\label{thm 5}
	Let the notations be defined as above. Then
	\begin{description}
		\item[(1)] the Singleton bound for the $\lambda$-constacyclic code $\mathcal{C}$ of length $m$ is
		$$b_{\mathrm{S}}(\mathcal{C}) = \omega\cdot p_{1}^{m_{1}}\cdots p_{r}^{m_{r}}+1,$$
		where $m_{i} = \mathrm{max}\{e_{i}-v_{p_{i}}(q^{\omega}-1),0\}$ for all $1 \leq i \leq r$; and
		\item[(2)] the arithmetic Singleton bound for $\mathcal{C}$ is
		$$b_{\mathrm{AS}}(\mathcal{C}) = \omega+1.$$
	\end{description}
\end{theorem}

\begin{proof}
	Since $f(X)$ is an irreducible polynomial of order $n$, its defining set is a $q$-cyclotomic coset modulo $n$, say,
	$$c_{n/q}(\gamma) = \{\gamma,\gamma q,\cdots,\gamma q^{\tau-1}\}$$
	for some $\gamma \in \mathbb{Z}/n\mathbb{Z}$ that is coprime to $n$.
	
	We first compute the arithmetic Singleton bound for $\mathcal{C}$. By Lemma \ref{lem 4} 
	$$c_{n/q}(\gamma) = \bigsqcup_{i=0}^{\omega-1}c_{n/q^{\omega}}(\gamma q^{i})$$
	is the coarsest MED representation of $c_{n/q}(\gamma)$, whose common difference is 
	$$d = \dfrac{n}{\frac{\tau}{\omega}} = \dfrac{n\omega}{\tau}.$$
	It follows from Theorem \ref{thm 4} that the arithmetic Singleton bound for $\mathcal{C}$ is given by
	$$b_{\mathrm{AS}}(\mathcal{C}) = \dfrac{\tau d}{n}+1 = \omega+1.$$
	
	Applying the lifting-the-exponent lemma, the degree of $f(X)$ is equal to
	$$\tau = \mathrm{ord}_{n}(q) = \omega\cdot p_{1}^{m_{1}}\cdots p_{r}^{m_{r}},$$
	where $m_{i} = \mathrm{max}\{e_{i}-v_{p_{i}}(q^{\omega}-1),0\}$ for each $1 \leq i \leq r$. Therefore the Singleton bound for $\mathcal{C}$ is given by
	$$b_{\mathrm{S}}(\mathcal{C}) = \omega\cdot p_{1}^{m_{1}}\cdots p_{r}^{m_{r}}+1.$$
\end{proof}

The criterion for the arithmetic Singleton bound and the classical Singleton bound coinciding is an immediate consequence of Theorem \ref{thm 5}.

\begin{corollary}\label{coro 2}
	The Singleton bound and the arithmetic Singleton bound for $\mathcal{C}$ coincide if and only if 
	$$e_{i} \leq v_{p_{i}}(q^{\omega}-1)$$
	for all $1 \leq i \leq r$.
\end{corollary}

On the other hand, also by Theorem \ref{thm 5} the difference of the arithmetic Singleton bound and the Singleton bound for $\mathcal{C}$ is
$$b_{\mathrm{S}}(\mathcal{C})-b_{\mathrm{AS}}(\mathcal{C}) = \omega(p_{1}^{m_{1}}\cdots p_{r}^{m_{r}}-1),$$
where $m_{i} = \mathrm{max}\{e_{i}-v_{p_{i}}(q^{\omega}-1),0\}$ for each $1 \leq i \leq r$. Notice that $v_{p_{i}}(q^{\omega}-1)$ is determined by $\mathrm{ord}_{\mathrm{rad}(n)}(q)$ up to a constant factor which is either $1$ or $2$, and thus is bounded when the $e_{i}$'s vary. Hence the difference of the two bounds grows dramatically as the $e_{i}$'s increase. Precisely, it can be formulated by the following asymptotic result.

\begin{corollary}
	Let $\ell$ be any prime divisor of $n$. Then we have
	$$\lim\limits_{v_{\ell}(n) \rightarrow \infty}(b_{\mathrm{S}}(\mathcal{C})-b_{\mathrm{AS}}(\mathcal{C})) = \infty.$$
\end{corollary}

We compute the arithmetic Singleton bound and compare it with the classical Singleton bound in a concrete example.

\begin{example}
	Let $q=7$ and $m=225 = 3^{2}\times5^{2}$. Let $\mathcal{C}$ be an irreducible cyclic code of length $m$ over $\mathbb{F}_{q}$, with the generator polynomial $f(X) \mid X^{m}-1$. Since
	$$X^{m}-1 = X^{225}-1 = \prod_{0\leq i,j \leq 2}\Phi_{3^{i}5^{j}}(X),$$
	where $\Phi_{3^{i}5^{j}}(X)$ denotes the $3^{i}5^{j}$-th cyclotomic polynomial, then $f(X)$ is a factor of $\Phi_{3^{i}5^{j}}(X)$ for some $0 \leq i,j \leq 2$. It is clear that the order of $f(X)$ is $3^{i}5^{j}$. Then
	\begin{equation*}
		\omega = \left\{
		\begin{array}{lcl}
			1, \quad \mathrm{if} \ j=0;\\
			4, \quad \mathrm{if} \ j=1,2,
		\end{array} \right.
	\end{equation*} 
	which implies that the arithmetic Singleton bound of $\mathcal{C}$ is
	\begin{equation*}
		b_{\mathrm{AS}}(\mathcal{C}) = \left\{
		\begin{array}{lcl}
			2, \quad \mathrm{if} \ j=0;\\
			5, \quad \mathrm{if} \ j=1,2.
		\end{array} \right.
	\end{equation*} 
	On the other hand, as
	$$v_{3}(q^{\omega}-1)=1$$
	and
	\begin{equation*}
		v_{5}(q^{\omega}-1) = \left\{
		\begin{array}{lcl}
			0, \quad \mathrm{if} \ \omega=1;\\
			2, \quad \mathrm{if} \ \omega=4,
		\end{array} \right.
	\end{equation*} 
	then the Singleton bound of $\mathcal{C}$ is 
	\begin{equation*}
		b_{\mathrm{S}}(\mathcal{C}) = \left\{
		\begin{array}{lcl}
			2, \quad \mathrm{if} \ i=0,1 \ \mathrm{and} \ j=0;\\
			4, \quad \mathrm{if} \ i=2 \ \mathrm{and} \ j=0;\\
			5, \quad \mathrm{if} \ i=0,1 \ \mathrm{and} \ j=1,2;\\
			13, \quad \mathrm{if} \ i=2 \ \mathrm{and} \ j=1,2.
		\end{array} \right.
	\end{equation*} 
We summarize these results into Table $1$.

\newpage

\begin{table}[hp]
\begin{center}
\begin{tabular}{c|c|c|c}
	\toprule
	$f(X)$ & $b_{\mathrm{S}}(\mathcal{C})$ & $b_{\mathrm{AS}}(\mathcal{C})$ & Do $b_{\mathrm{S}}(\mathcal{C})$ and $b_{\mathrm{AS}}(\mathcal{C})$ coincide?\\
				\midrule
				a factor of $\Phi_{1}(X)$ & $2$ & $2$ & yes\\
				a factor of $\Phi_{3}(X)$ & $2$ & $2$ & yes\\
				a factor of $\Phi_{9}(X)$ & $4$ & $2$ & no\\
				a factor of $\Phi_{5}(X)$ & $5$ & $5$ & yes\\
				a factor of $\Phi_{15}(X)$ & $5$ & $5$ & yes\\
				a factor of $\Phi_{45}(X)$ & $13$ & $5$ & no\\
				a factor of $\Phi_{25}(X)$ & $5$ & $5$ & yes\\
				a factor of $\Phi_{75}(X)$ & $5$ & $5$ & yes\\
				a factor of $\Phi_{225}(X)$ & $13$ & $5$ & no\\
				\bottomrule
\end{tabular}
\end{center}
\caption{A comparison of Singleton bounds and arithmetic Singleton bounds in some examples}
\end{table}
\end{example}

Finally we remark that in the case that $f(X)$ is irreducible over $\mathbb{F}_{q}$ a stronger conclusion can be obtained. Let $d_{0} = \frac{n\omega}{\tau}$ be the smallest integer in $\Sigma_{\gamma}$. By the proof of Theorem \ref{thm 5} we have
$$d_{0} = p_{1}^{v_{1}}\cdots p_{r}^{v_{r}},$$
where $v_{i} = \mathrm{min}\{e_{i},v_{p_{i}}(q^{\omega}-1)\}$ for all $1 \leq i \leq r$. Note that $1 \leq v_{i} \leq v_{p_{i}}(q^{\omega}-1)$, and if further $p_{i} = 2$ and $e_{i} \geq 3$ then $2 \leq v_{i} \leq v_{p_{i}}(q^{\omega}-1)$. Then by definition $\omega$ is the smallest positive integer such that 
$$\gamma q^{\omega} \equiv \gamma \pmod{d_{0}}.$$
It follows that the $q$-cyclotomic coset $c_{d_{0}/q}(\gamma)$ modulo $d_{0}$ is
$$c_{d_{0}/q}(\gamma) = \{\gamma,\gamma q,\cdots,\gamma q^{\omega-1}\},$$
which induces a polynomial 
$$f_{0}(X) = \prod_{i=0}^{\omega-1}(X-\zeta_{d_{0}}^{\gamma q^{i}})$$
over $\mathbb{F}_{q}$.

\begin{proposition}\label{prop 4}
	The polynomial $f(X)$ has the same Hamming weight as $f_{0}(X)$.
\end{proposition}

\begin{proof}
	The coarsest MED representation of $c_{n/q}(\gamma)$ is
	$$c_{n/q}(\gamma) = \bigsqcup_{i=0}^{\omega-1}c_{n/q^{\omega}}(\gamma q^{i}),$$
	which induces the factorization
	$$f(X) = \prod_{i=0}^{\omega-1}(X^{\frac{n}{d_{0}}}-\zeta_{d_{0}}^{\gamma q^{i}})$$
	of $f(X)$. Then one has $f(X) = f_{0}(X^{\frac{n}{d_{0}}})$ and consequently
	$$\mathrm{wt}_{H}(f) = \mathrm{wt}_{H}(f_{0}).$$
\end{proof}

As $d_{0} = p_{1}^{v_{1}}\cdots p_{r}^{v_{r}}$, where $v_{i} = \mathrm{min}\{e_{i},v_{p_{i}}(q^{\omega}-1)\}$ for any $1 \leq i \leq r$, Proposition \ref{prop 4} indicates that in order to obtain the Hamming weights of irreducible polynomials of order $p_{1}^{e_{1}}\cdots p_{r}^{e_{r}}$ for all $e_{1},\cdots,e_{r} \geq 1$, one only needs to compute the Hamming weights of a finite number of polynomials.

\section{The comparison of the arithmetic Singleton bound and the BCH bound}\label{sec 5}
In this section we compare the arithmetic Singleton bound and the BCH bound for simple-root cyclic codes over $\mathbb{F}_{q}$. In particular, an equivalent condition on that the arithmetic Singleton bound coincides with the BCH bound, in which case the Hamming distance can be obtained directly, is presented. As corollaries, an equivalent condition for an irreducible cyclic code to have Hamming distance $2$, and a sufficient condition for it to have Hamming distance $3$ are given.

Let 
$$\mathcal{C} = (f(X)) \subseteq \mathbb{F}_{q}[X]/(X^{n}-1)$$
be a cyclic code of length $n$ over $\mathbb{F}_{q}$, where $\mathrm{gcd}(n,q)=1$ and $f(X) \mid X^{n}-1$ is the generator polynomial of $\mathcal{C}$. Without loss of generality we assume that $f(X) \neq X^{n}-1$. And note that if $\mathrm{ord}(f) = r < n$ then 
$$X^{r}-1 \in (f(X)) = \mathcal{C}$$
and consequently $d_{H}(\mathcal{C}) = 2$. Therefore we also assume that $\mathrm{ord}(f) = n$ in the remainder of this section.

Let $\overline{\Sigma}_{f} = d\mathbb{Z}/n\mathbb{Z}$ be the additive stabilizer of the defining set $T_{f}$ of $f(X)$, and let
$$T_{f} = \bigsqcup_{i=1}^{s}E_{i} = \bigsqcup_{i=1}^{s}(\gamma_{i}+d\mathbb{Z}/n\mathbb{Z})$$
be the coarsest MED representation of $T_{f}$. We denote by
$$\widetilde{T}_{f} = \{\gamma_{1},\cdots,\gamma_{s}\} \pmod{d} \subseteq \mathbb{Z}/d\mathbb{Z}$$
the image of the set $\{\gamma_{1},\cdots,\gamma_{s}\}$ in the residue class ring $\mathbb{Z}/d\mathbb{Z}$. Clearly this set $\widetilde{T}_{f}$ is independent of the choice of the representatives $\gamma_{i}$ of $E_{i}$ for $1 \leq i \leq s$.

\begin{lemma}\label{lem 6}
	For a positive integer $\delta \leq d-1$, $T_{f}$ contains $\delta$ consecutive integers if and only if $\widetilde{T}_{f}$ does.
\end{lemma}

\begin{proof}
	If $\widetilde{T}_{f}$ contains $\delta$ consecutive integers, it is obvious that so does $T_{f}$. Conversely, suppose that $T_{f}$ contains $\delta$ consecutive integers, say,
	$$\{b,b+1,\cdots,b+\delta-1\} \subseteq T_{f}.$$
	Since $\overline{\Sigma}_{f} = d\mathbb{Z}/n\mathbb{Z}$, $b+i+d\mathbb{Z}/n\mathbb{Z} \subseteq T_{f}$ for all $0 \leq i \leq \delta-1$. We claim that for any $0 \leq i < j \leq \delta-1$ the cosets $b+i+d\mathbb{Z}/n\mathbb{Z}$ and $b+j+d\mathbb{Z}/n\mathbb{Z}$ are distinct. Otherwise,
	$$b+i+d\mathbb{Z}/n\mathbb{Z} = b+j+d\mathbb{Z}/n\mathbb{Z}$$
	yields $j-i \in d\mathbb{Z}/n\mathbb{Z}$, which contradicts the condition $\delta \leq d-1$. Hence one has
	$$\bigsqcup_{i=0}^{\delta-1}(b+i+d\mathbb{Z}/n\mathbb{Z}) \subseteq T_{f},$$
	which implies that $\{b,b+1,\cdots,b+\delta-1\} \subseteq \widetilde{T}_{f}$.
\end{proof}

For any integer $a$ coprime to $n$, let $f_{a}(X)$ be the factor of $X^{n}-1$ with defining set
$$a^{-1}T_{f} = \{a^{-1}\gamma \ | \ \gamma \in T_{f}\},$$
and let $\mathcal{C}_{a}$ be the cyclic code of length $n$ generated by $f_{a}(X)$. The conclusion below is well-known. (See for instance Section $4.3$ in \cite{Huffman}.)

\begin{lemma}\label{lem 7}
	The cyclic codes $\mathcal{C}$ and $\mathcal{C}_{a}$ are permutation equivalent via the map
	$$\mu_{a}: \mathcal{C} \rightarrow \mathcal{C}_{a}; \ h(X) \mapsto h(X^{a}) \pmod{X^{n}-1}.$$
	In particular, $\mathcal{C}$ and $\mathcal{C}_{a}$ have the same Hamming distance.
\end{lemma}

Combining Lemma \ref{lem 6} and \ref{lem 7}, we obtain the following theorem which compares the arithmetic Singleton bound and the BCH bound for $\mathcal{C}$.

\begin{theorem}
	Let the notations be defined as above. Set $\delta$ to be the largest integer for which there exist a positive integer $a$ coprime to $n$ and an integer $b$ such that 
	$$\{b,b+a,\cdots,b+(\delta-1)a\} \subseteq \widetilde{T}_{f}.$$
	Then the arithmetic Singleton bound for $\mathcal{C}$ is $| \widetilde{T}_{f} |+1$, and the BCH bound for $\mathcal{C}$ is $\delta+1$. Therefore the Hamming distance of $\mathcal{C}$ satisfies
	$$\delta+1 \leq d_{H}(\mathcal{C}) \leq | \widetilde{T}_{f} |+1.$$
\end{theorem}

\begin{corollary}\label{coro 5}
	The arithmetic Singleton bound and the BCH bound for $\mathcal{C}$ coincide if and only if $T_{f}$ has the form 
	$$T_{f} = (b+H)\sqcup(b+a+H)\sqcup\cdots\sqcup(b+(\delta-1)a+H),$$
	where $b$ is an integer, $a$ is a positive integer coprime to $n$, and $\delta$ is a positive integer less than $d$. 
If this is the case, then the Hamming distance of $\mathcal{C}$ is $\delta+1$. 
\end{corollary}

If $\mathcal{C}$ is an irreducible cyclic code, a criterion for the Hamming distance of $\mathcal{C}$ to be $2$, and a 
sufficient condition for it to be $3$ are given as follows, which are direct consequences of Corollary \ref{coro 5}.

\begin{corollary}\label{coro 6}
	Assume that $\mathcal{C}$ is an irreducible cyclic code.
	\begin{description}
		\item[(1)] The Hamming distance of $\mathcal{C}$ is $2$ if and only if $\mathrm{gcd}(n,q-1)> 1$.
		\item[(2)] If $\mathrm{gcd}(n,q-1)=1$ and $\omega =2$, then the Hamming distance of $\mathcal{C}$ is $3$.
	\end{description}
\end{corollary}

Finally, we exhibit Corollary \ref{coro 5} and \ref{coro 6} with the next two examples.

\begin{example}
	\begin{itemize}
		\item[(1)] Let $q = 5$ and $n = 2^{e_{1}}3^{e_{2}}$ for $e_{1},e_{2} \in \mathbb{N}$. Let 
		$$\mathcal{C} = (f(X)) \subseteq \mathbb{F}_{q}[X]/(X^{n}-1)$$
		be an irreducible cyclic code, with $\mathrm{ord}(f) = n$. Then the Hamming distance of $\mathcal{C}$ is 
		\begin{equation*}
			d_{H}(\mathcal{C}) = \left\{
			\begin{array}{lcl}
				2, \ \mathrm{if} \ e_{1}\geq 1;\\
				3, \ \mathrm{if} \ e_{1}=0.
			\end{array} \right.
		\end{equation*}
		
		\item[(2)] Let $q = 11$ and $n = 35$. Write 
		$$H = 5\mathbb{Z}/35\mathbb{Z} \subseteq \mathbb{Z}/35\mathbb{Z}.$$
		The subset 
		$$T = (1+H)\sqcup(2+H)\sqcup(3+H)\sqcup(4+H)$$
		gives rise to a polynomial $f_{T}(X)$ over $\mathbb{F}_{q}$. Then the cyclic code
		$$\mathcal{C}_{T} = (f(X)) \subseteq \mathbb{F}_{q}[X]/(X^{n}-1)$$
		has arithmetic Singleton bound and BCH bound both equaling to $5$, hence
		$$d_{H}(\mathcal{C}) = 5.$$
		On the other hand, the Singleton bound for $\mathcal{C}_{T}$ is $29$.
	\end{itemize}
\end{example}

\section{Conclusion}
In this paper, we investigated how the arithmetic structure of a simple-root constacyclic code constrains its Hamming distance and introduced an explicit quantitative measure of this phenomenon. Our main contributions can be summarized as follows.

We generalize the notion of a MED representation of a cyclotomic coset, which is first introduced in \cite{Zhu3}, to the case of the defining set of any squarefree polynomial over $\mathbb{F}_{q}$, and show that the results on the structure and basic properties of the space of MED representations of a cyclotomic coset have their natural counterparts in the general case. Given a simple-root constacyclic code
$$\mathcal{C} = (f(X)) \subseteq \mathbb{F}_{q}[X]/(X^{m}-\lambda),$$
where $f(X) \mid X^{m}-\lambda$ is the generator polynomial of $\mathcal{C}$, with order $n$ and defining set $T_{f}$, the MED representations of $T_{f}$ leads to a family of upper bounds on the Hamming distance of $\mathcal{C}$, among which the weakest one coincides with the Singleton bound, while the strongest one is defined to be the arithmetic Singleton bound for this code. The arithmetic Singleton bound to some extent measures the restriction on the Hamming distance of $\mathcal{C}$ imposed by its arithmetic structure.

Clearly the arithmetic Singleton bound is always stronger than the Singleton bound, and their difference is governed by the additive stabilizer $\overline{\Sigma}_{f}$ of the defining set of $f(X)$. In particular, the arithmetic Singleton bound is strictly stronger than the Singleton bound if and only if $\overline{\Sigma}_{f}$ is nontrivial, which happens for infinitely many families of constacyclic codes. For instance, as the MED representations of a cyclotomic coset can be completed determined, the arithmetic Singleton bound for an irreducible constacyclic code, together with the gap between it and the classical Singleton bound, can be computed concretely. If the order $n$ of its generator polynomial has the prime factorization
$$n = p_{1}^{e_{1}}\cdots p_{r}^{e_{r}},$$
then the difference between the arithmetic Singleton bound and the Singleton bound grows dramatically as the $e_{i}$'s grow. In fact, if any $e_{i}$ tends to infinity then so does the difference.

We propose several problems for future research:

\begin{itemize}
	\item[(1)] The additive stabilizer $\overline{\Sigma}_{f}$, or equivalently the common difference set $\Sigma_{f}$, associated to the defining set $T_{f}$ of $f(X)$ determines the MED representations of $T_{f}$, and consequently governs the arithmetic Singleton bound for the constacyclic code generated by $f(X)$. Can we give an explicit formula to compute $\overline{\Sigma}_{f}$ and $\Sigma_{f}$?
	\item[(2)] For an irreducible constacyclic code $\mathcal{C}$ with generator polynomial $f(X)$, we see that the arithmetic Singleton bound for $\mathcal{C}$ is determined by $\mathrm{ord}_{\mathrm{rad}(n)}(q)$ up to a constant factor, where $n$ is the order of $f(X)$. Therefore if fixing the prime divisors of $n$, the arithmetic Singleton bound is bounded when $n$ grows. To what extent is this fact true in the general case?
	\item[(3)] Extend the MED representation framework to repeated-root constacyclic codes.
\end{itemize}

\section*{Acknowledgment}
The first author was supported by Basic Research Program Young Scientists Guidance Project of Guizhou Province (QN[2025]186).

\section*{Data availability}
Data sharing not applicable to this article as no datasets were generated or analysed during the current study.

\section*{Declaration of competing interest}
The authors declare that we have no known competing financial interests or personal relationships 
that could be perceived to influence the work reported in this paper.


\begin{thebibliography}{99}
\bibitem{Alderson} T. Alderson, On the weights of general MDS codes, IEEE Transactions on Information Theory, vol. 66, no. 9, 2020.

\bibitem{Dinh} Hai Q. Dinh, Xiaoqiang Wang, Hongwei Liu, Woraphon Yamaka, Hamming distances of constacyclic codes of length 3ps and optimal codes with respect to the Griesmer and Singleton bounds, Finite Fields and Their Applications, 70(2021), 101794.

\bibitem{Huffman} W. Huffman, V. Pless, Fundamentals of Error-Correcting Codes, Cambridge: Cambridge University Press, 2003.
	
\bibitem{MacWilliams} F. MacWilliams, N. Sloane, The Theory of Error-Correcting Codes, North-Holland, 1977.
	  
\bibitem{Nezami} S. Nezami, Leme Do Khat (in English: Lifting The Exponent Lemma), pulished on Oct., 2009. 
	
\bibitem{Wan} Z. Wan, Lectures on Finite Fields and Galois Rings. World Scientific, Singapore, 2003.  
	
\bibitem{Zhu} L. Zhu, J. Liu, H. Wu, Explicit representatives and sizes of cyclotomic cosets and their application to cyclic codes over finite fields, Finite Fields and Their Applications, 111(2026), 102761.
	
\bibitem{Zhu3} L. Zhu, J. Zhou, J. Liu, H. Wu, The Multiple Equal-Difference Structure of Cyclotomic Cosets, arXiv: 2501.03516.	
	
\bibitem{Zhu2} L. Zhu, H. Wu, The Minimal Binomial Multiples of Polynomials over Finite Fields, arXiv: 2510. 18624v1.
\end{thebibliography}
\end{document}